\documentclass[conference]{IEEEtran}
\usepackage{color,graphicx}
\usepackage{amsmath}
\usepackage{amssymb}
\usepackage{amsthm}
\usepackage{amsfonts}
\usepackage[ruled]{algorithm2e}
\usepackage{enumerate}
\usepackage{multicol}
\usepackage{array}
\usepackage{url}
\usepackage{multirow}
\usepackage{subfigure}

\newtheorem{theorem}{Theorem}

\newtheorem{myDef}{Definition}
\newtheorem{example}{Example}

\begin{document}

\title{On the Difficulty of Inserting Trojans in Reversible Computing Architectures}

\author{\IEEEauthorblockN{Xiaotong Cui\IEEEauthorrefmark{1},
Samah Saeed\IEEEauthorrefmark{2},
Alwin Zulehner\IEEEauthorrefmark{3},
Robert Wille\IEEEauthorrefmark{3},
Rolf Drechsler\IEEEauthorrefmark{4},
Kaijie Wu\IEEEauthorrefmark{5},
Ramesh Karri\IEEEauthorrefmark{5},}
\IEEEauthorblockA{\IEEEauthorrefmark{1}Chongqing University, China}
\IEEEauthorblockA{\IEEEauthorrefmark{2}University of Washington Tacoma, USA}
\IEEEauthorblockA{\IEEEauthorrefmark{3}Johannes Kepler University Linz, Austria}
\IEEEauthorblockA{\IEEEauthorrefmark{4}University of Bremen, German}
\IEEEauthorblockA{\IEEEauthorrefmark{4}New York University, USA}}

\maketitle

\begin{abstract}
Fabrication-less design houses outsource their designs to 3rd party foundries to lower fabrication cost. However, this creates opportunities for a rogue in the foundry to introduce hardware Trojans, which stay inactive most of the time and cause unintended consequences to the system when triggered. Hardware Trojans in traditional CMOS-based circuits have been studied and Design-for-Trust (DFT) techniques have been proposed to detect them.

Different from traditional circuits in many ways, reversible circuits implement one-to-one, bijective input/output mappings. We will investigate the security implications of reversible circuits with a particular focus on susceptibility to hardware Trojans. We will consider inherently reversible circuits and non-reversible functions embedded in reversible circuits.
\end{abstract}

\IEEEpeerreviewmaketitle


\section{Introduction}

Most established computing paradigms are not reversible including the building block operations such as NAND/NOR that they use. While it is possible to infer the inputs when a NAND  outputs a 0 (then, both inputs are 1), it is not possible to unambiguously infer the inputs if the NAND outputs 1.

Alternative computing paradigms are gaining interest. \emph{Reversible computations} are bijective $n$-input $n$-output functions that map each input combination to a unique output combination. Reversibility is useful in  implementing quantum computing architectures, since quantum computations are inherently reversible~\cite{2000:QCQ:544199,BBC+:95,DBLP:conf/rc/NiemannBCJW15}. In fact, many components of quantum computers such as the database in case of Grover's Search \cite{Gro:96} or modular exponentiation in case of Shor's algorithm \cite{Sho:94}) are reversible. Hence, researchers built a reversible circuit first (using methods e.g. reviewed in \cite{SAZS:2011, DBLP:conf/ismvl/DrechslerW11}), which are afterwards mapped into a corresponding quantum circuit (using methods proposed in \cite{BBC+:95,GWDD:2009b}).

Besides, reversible logic continues to grow and show promise in low power computing, ~\cite{Landauer61,Ben:73,BAP+:2012}, adiabatic computing~\cite{363692}, circuit verification~\cite{DBLP:conf/date/AmaruGW16}, and optical computing~\cite{Cuykendall:87,roy2011all}.

Reversible circuits differ from conventional circuits in many ways. In order to implement reversibility, fan-out and feedback are not allowed and each circuit is realized as a cascade of reversible gates. 
In this paper, we will study the supply chain security of reversible circuits with a focus on hardware Trojans.
Effects of different types of Trojans are studied and simple yet effective defenses are proposed. 

The paper is organized as follows. Section \ref{sec:rev_backgroud} introduces reversible logic. Section \ref{sec:related_work} presents the threat model in reversible circuits. In Section \ref{sec:security}, we conduct studies on small-sized hardware Trojans in reversible circuits. Based on the characteristics of reversible circuits, simple yet effective methods are proposed to disable inserted hardware Trojans in Section \ref{sec:proposed}. Experiments on effectiveness of the proposed methods are presented in Section \ref{sec:experiments} and Section \ref{sec:conclusion} concludes the paper.

\section{Review on Reversible Logic}
\label{sec:rev_backgroud}

\begin{myDef}
A Boolean function $f:\mathbb{B}^n \rightarrow \mathbb{B}^n$ is reversible if it maps each input combination to a unique output combination.
\end{myDef}

A  reversible boolean function with $n$ inputs can be viewed as a mapping of a $Set=\{0, 1, \cdots, 2^n-1\}$ to itself. A reversible function can be represented as a truth table or a permutation matrix as shown in Fig. \ref{fig:rev}. Each input uniquely maps to an output and vice versa, i.e., the mapping is a bijection.

\begin{figure}[htbp]
\centering
\includegraphics[width=0.32\textwidth]{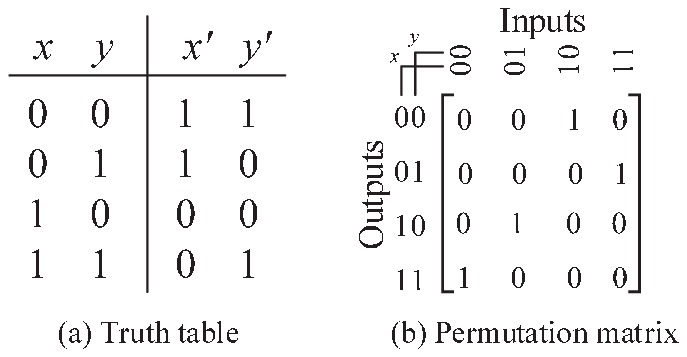}
\caption{Two representations of a reversible function.}
\label{fig:rev}
\end{figure}

Reversible circuits synthesized from reversible functions typically are cascades of reversible gates. In addition to components that are naturally reversible such as wire, Not, and Swap gates, Fredkin, Toffoli and Feynman are frequently used in the reversible/quantum computing literature \cite{taha2016fundamentals}.
\begin{myDef}
A $k$ input, $k$ output gate is reversible if it realizes a reversible function.
\end{myDef}
Toffoli gate will be considered in this paper. A $k$-input Toffoli gate has $k-1$ control lines and 1 target line, where the control lines can be either positive or negative.
\begin{myDef}
A $k\times k$ Toffoli gate passes k-1 control lines unchanged, and inverts the target line if values of all its positive (negative) control lines are 1 (0).
\end{myDef}

For simplicity, only positive control lines are considered in  this paper. Consider the general Toffoli gate $TOFk(x_1, x_2, \cdots, x_k)=(x'_1, x'_2, \cdots, x'_k)$ where the $k^{th}$ line is the target line, then
\begin{equation*}
\begin{split}
&x'_i = x_i (i<k),\\
&x'_k = x_1x_2\cdots x_{k-1}\oplus x_k.
\end{split}
\end{equation*}
$TOF1(x_1)$ is a NOT gate, and $TOF2(x_1, x_2)$ is a controlled-NOT gate (CNOT). Symbols of $TOF1(x_1)$, $TOF2(x_1, x_2)$ and $TOF3(x_1, x_2, x_3)$ are shown in Fig. \ref{fig:toffo}.

\begin{figure}[htbp]
\centering
\includegraphics[width=0.4\textwidth]{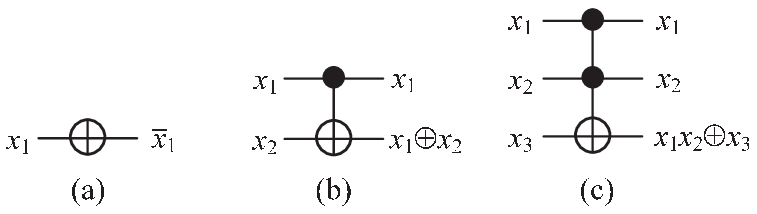}
\caption{Example reversible gates. (a) NOT, (b) C-NOT, (3) Toffoli.}
\label{fig:toffo}
\end{figure}

Arbitrary functions can be implemented as reversible circuits as cascades of reversible gates. Many approaches have been proposed to synthesize reversible circuits \cite{SPMH:2002,MillerMD03,YSHP:2005,GAJ:2006,wille2009bdd,SWH+:2012, ZulehnerW17Emb}. Ancillary inputs and garbage outputs are necessary when embedding non-reversible function into a reversible structure.

\begin{figure}[htbp]
\centering
\includegraphics[width=0.32\textwidth]{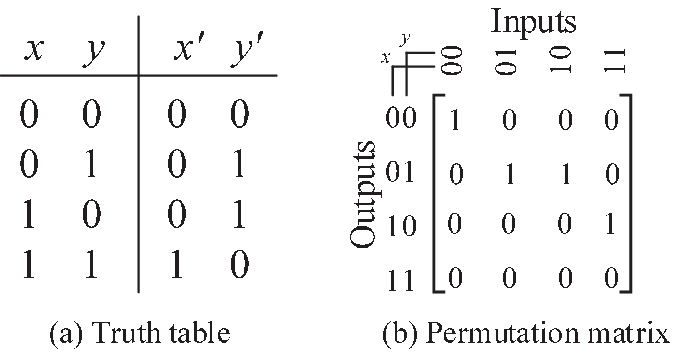}
\caption{Truth table and permutation matrix of a half adder.}
\label{fig:halfadd}
\end{figure}

\begin{example}
The half adder truth table is shown in Fig. \ref{fig:halfadd}(a). Input combinations 01 and 10 map to the same output 01. Fig. \ref{fig:halfadd}(b) shows its permutation matrix. Non-reversibility can be seen in the second row (output 01) of the matrix, which has two 1s in the row. The entries in the second and third columns represent input combinations 01 and 10. All 0s in the last row of the matrix violates reversibility, because no input combination is mapped to this output 11.
\end{example}

To make the half adder reversible, the first step is to add garbage output(s) to make sure each input pattern has a corresponding unique output pattern. In a half adder, there are two input patterns that map to one output pattern, so adding one garbage output bit is enough to create a unique input-output mapping. Next, we should make
sure that the number of inputs and the number of outputs are equal. One ancillary (constant) input is added to make the number of inputs and outputs same, as shown in TABLE \ref{tab:rehalf_1}. Any assignment to the garbage output \emph{g} and the ancillary input \emph{a} is a valid assignment as long as it creates a reversible function. Most prior work assigned constant values to ancillary inputs as shown in TABLE \ref{tab:rehalf_2}.

\begin{table}
\caption{Embedding the half adder.}
\label{tab:halfrev}
\centering
\subtable[Adding ancillary inputs and garbage outputs]{
       \begin{tabular}{ccc|ccc}
        $x$ & $y$ & $a$ & $x'$ & $y'$ & $g$ \\\hline
        0   &  0  &  0  &  0   &  0   &  *  \\
        0   &  0  &  1  &  $\cdot$   &  $\cdot$   &  $\cdot$  \\
        0   &  1  &  0  &  0   &  1   &  *  \\
        0   &  1  &  1  &  $\cdot$   &  $\cdot$   &  $\cdot$  \\
        1   &  0  &  0  &  0   &  1   &  *  \\
        1   &  0  &  1  &  $\cdot$   &  $\cdot$   &  $\cdot$  \\
        1   &  1  &  0  &  1   &  0   &  *  \\
        1   &  1  &  1  &  $\cdot$   &  $\cdot$   &  $\cdot$  \\
       \end{tabular}
       \label{tab:rehalf_1}
}
\quad
\subtable[One possible embedding]{
       \begin{tabular}{ccc|ccc}
        $x$ & $y$ & $a$ & $x'$ & $y'$ & $g$ \\\hline
        \textbf{0}   &  \textbf{0}  &  \textbf{0}  &  \textbf{0}   &  \textbf{0}   &  \textbf{0}  \\
        0   &  0  &  1  &  0   &  0   &  1  \\
        \textbf{0}   &  \textbf{1}  &  \textbf{0}  &  \textbf{0}   &  \textbf{1}   &  \textbf{0}  \\
        0   &  1  &  1  &  1   &  0   &  0  \\
        \textbf{1}   &  \textbf{0}  &  \textbf{0}  &  \textbf{0}   &  \textbf{1}   &  \textbf{1}  \\
        1   &  0  &  1  &  1   &  1   &  0  \\
        \textbf{1}   &  \textbf{1}  &  \textbf{0}  &  \textbf{1}   &  \textbf{0}   &  \textbf{1}  \\
        1   &  1  &  1  &  1   &  1   &  1  \\
       \end{tabular}
       \label{tab:rehalf_2}
}
\end{table}

\section{The Threat Model}
\label{sec:related_work}
\subsection{Hardware Trojans}
Globalization of the integrated circuit (IC) design flow has raised security concerns including malicious circuit modifications, referred as \emph{Hardware Trojans}. A hardware Trojan remains inactive and poses no threat to its host circuit until it is triggered. When triggered it will alter the original function \cite{Jin2009Hardware,tehranipoor2010survey}. Hardware Trojans are hard to detect since they are stealthy: 1) they are designed to be triggered under extremely rare conditions and 2) they are usually much smaller than their host circuits. Various Design-for-Trust techniques have been developed to detect and prevent hardware Trojans and/or recover from faults created by them \cite{xiao2013bisa,cui2014high}. One set of techniques increase the switching probability of the normally inactive nets  and hence facilitating the triggering and detecting processes \cite{salmani2012anovel}. Another set of techniques detect hardware Trojans by enhancing and monitoring the side-channel effects \cite{wang2008hardware,jin2008hardware, 7390050}. However, all of these techniques focus on CMOS-based designs that are inherently non-reversible.

\subsection{Trojans vs. Faults in Reversible Logic}
Hardware Trojans are to some extent similar to faults as they all affect the outputs when triggered.
Extending from this point, a hardware Trojan may be categorized as a new type of faults where the host circuit contains extra gates.
However, the key difference between faults and hardware Trojans is their origins: Faults originate from imperfect fabrication process or environment disturbance, while hardware Trojans are well-designed and carefully inserted by attackers with a detailed knowledge of the host circuit.
From the point of circuit testing, a common practice is to focus on a limited set of fault models and assume a single-gate failure \cite{patel2004fault,ramasamy2004fault, hayes2004testing, wille2011atpg}.
On the other hand, the size, structure, and functionality of a Trojan can be unknown and unforeseeable, which makes test patterns for Trojan detection much harder to generate.

Further, the excellent controllability and reversibility of reversible circuits makes testing for faults much easier. In a reversible circuit, reversibility exists not only between primary inputs and primary outputs but also between inputs and outputs of any internal gates. The inputs to any internal gate can be controlled easily from the primary inputs and the outputs of any internal gate can be observed at the primary outputs.
The excellent controllability and reversibility, however, benefit Trojan attackers as well. This is because all the nodes will have the same switching probability if random inputs are applied. As a result, a Trojan can be inserted at any position in a reversible circuit and receive the same triggering probability, which significantly increases the search space for detection. This is not the case in CMOS-based circuits where the switching probabilities of nodes vary significantly and inactive nodes are considered the best candidates for Trojan insertion \cite{salmani2012anovel}.

\subsection{The Threat Model}

Consider an attacker in the foundry who tries to insert a hardware Trojan during fabrication. The Trojan impacts the output of the host design when triggered. Fig.~\ref{fig:chain1} shows the supply chain of integrated circuits.

\begin{figure}[htbp]
\centering
\includegraphics[width=0.42\textwidth]{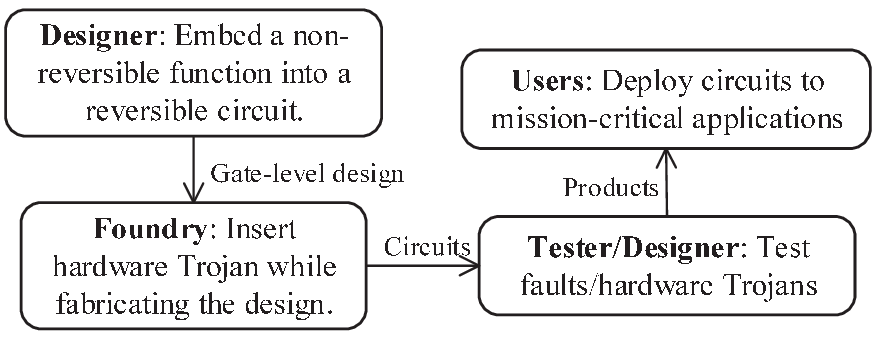}
\caption{The IC design flow. The attacker is in the foundry.}
\label{fig:chain1}
\end{figure}

The Trojan can be a simple $k$-Toffoli gate or complex one consisting of multiple $k$-Toffoli gates. $k$ can be different values for different gates. The hardware Trojan maintains reversibility since it uses Toffoli gates.

\section{Hardware Trojans in Reversible Circuits}
\label{sec:security}
Small hardware Trojans are always favored due to their minimal side-channel impacts on the host circuit including power, dimensions and delay. We start our analysis from small-sized Trojans to a general case.

\subsection{Single-Gate Trojans}

Consider a reversible circuit with $n$ primary inputs/outputs and a $k$-Toffoli Trojan gate. This Trojan will pass its input unchanged if at least one of its control lines is set to 0 and will affect the original circuit function when it is triggered, i.e., when all its control lines are set to 1.

\begin{theorem} A single-gate Trojan in a reversible circuit can be activated by at least two input patterns.
\end{theorem}
\begin{proof}
When the Trojan consists of n-1 control lines and 1 target line, it can be triggered by assigning each control line to  1, and the target line to either 1 or 0. So there are two input patterns that can activate the Trojan. When the number of control lines is $k-1$ ($k\leq n$), there are $2^{n-k+1}$ patterns that can trigger the Trojan. As a result, there are at least 2  patterns that can activate the Trojan.
\end{proof}

Since the tester does not know which line will be the target line in the Trojan, between the two triggering patterns, the pattern that applies 1 to all lines is more effective in detecting a single-gate Trojan.
The pattern, which needs to apply 0 to the target line of a Trojan, is less effective since the tester must try out all $n$ lines to guarantee the triggering. This results in the a set of one-cold patterns where each of the $n$ inputs takes turn to be 0.
We hence refer the former as the \emph{All-1} pattern, and the latter as the set of \emph{One-Cold} patterns.
According to the following analysis, the \emph{All-1} pattern is more efficient in triggering a single-gate Trojan, while the set of \emph{One-Cold} patterns are useful in triggering any Trojans consisting of 4 gates.

\underline{The \emph{ALL-1} pattern}: Since most of the times a Trojan is inserted somewhere in the middle of the circuit as shown in Fig. \ref{fig:rtrojan}, applying an \emph{All-1} triggering pattern at $A$ and observing the error at output are necessary.

\begin{figure}[htbp]
\centering
\includegraphics[width=0.4\textwidth]{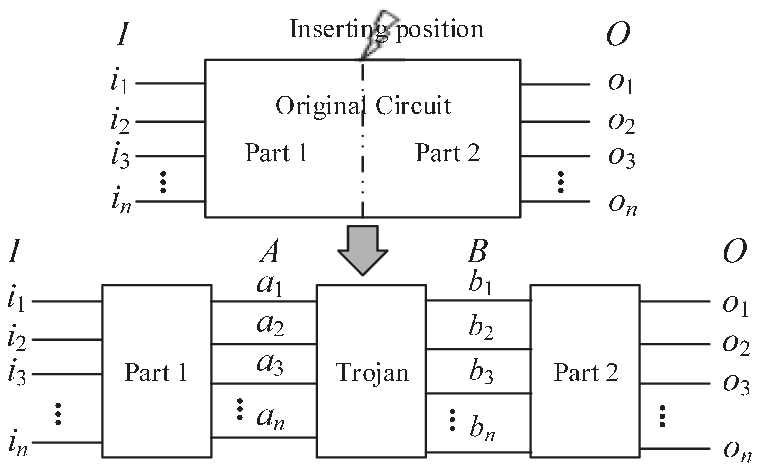}
\caption{Reversible circuit with Trojan insertion.}
\label{fig:rtrojan}
\end{figure}

\begin{theorem}
\label{theorem:t2}
It is sufficient to apply at most $m+1$ patterns at $I$ to trigger a single-gate Trojan inserted at ANY position in an $m$-gate reversible circuit, and observe the Trojan effects at $O$.
\end{theorem}
\begin{proof}
Controllability:
In order to apply the \emph{All-1} pattern at Trojan input $A$, the tester needs to find the corresponding pattern at circuit input $I$.
Luckily, due to the reversibility of the circuit and each sub-circuit (like Part 1 in the figure), there is always a pattern at $I$ that maps to the triggering pattern at $A$.

Observability:
Also due to the reversibility of Part 2, the effect induced by a triggered Trojan will always propagate an unexpected output observed at $O$.

The number of patterns: The pattern at $I$, however, depends on the function of Part 1 in Fig.~\ref{fig:rtrojan}, which in turn depends on the position of the Trojan.
Since in an $m$-gate reversible host circuit there are $m+1$ positions for inserting a Trojan, a $m+1$ input patterns at $I$ will guarantee the triggering of a single-gate Trojan inserted at ANY position.

These $m+1$ input patterns are sufficient to detect a single-gate Trojan.
\end{proof}

For easy reference, we still use the \emph{All-1} pattern (at $A$) to represent the $m+1$ patterns at $I$.

\underline{\emph{One-Cold} Patterns}:
\begin{theorem}
It is sufficient to apply at most $n\times(m+1)$ patterns at $I$ to trigger a single-gate Trojan inserted at ANY position in a $m$-gate reversible circuit, and observe its effect at $O$.
\end{theorem}
The proof is similar to Theorem \ref{theorem:t2}. Since each set of \emph{One-Cold} patterns consists of $n$ patterns for an $n$-line reversible circuit, the total number of inputs patterns that guarantee the detection of any single-gate Trojan is $n\times(m+1)$.
Similar to the \emph{All-1} pattern, we use the set of \emph{One-Cold} patterns at $A$ to represent the $n\times(m+1)$ patterns at $I$.

\subsection{Two-Gate Trojans}

For a two-gate Trojan, 2 cases arise.

Case 1: The two gates have different target lines as shown in Fig.~\ref{fig:2gate}(a). Triggering one of the gates triggers the Trojan. Both the \emph{All-1} pattern and the set of \emph{One-Cold} patterns can detect such two-gate Trojans. \emph{Thus, detection of this two-gate Trojan only needs partial triggering.}

\begin{figure}[htbp]
\centering
\includegraphics[width=0.42\textwidth]{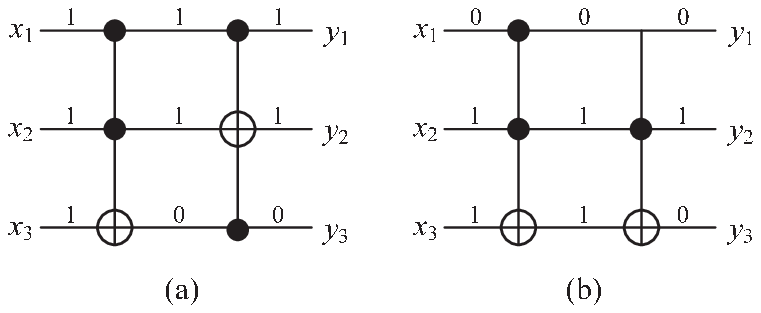}
\caption{Two-gate Trojan. (a) Two gates have different target lines; (b) Two gates have the same target line but different control lines.}
\label{fig:2gate}
\end{figure}

Case 2: The two gates have the same target line but different control lines, as shown in Fig.~\ref{fig:2gate}(b).
\begin{theorem}
At least one pattern from the set of \emph{One-Cold} patterns can detect the two-gate Trojan in Case 2 where both gates control the same target line.
\end{theorem}
\begin{proof}
Since the two gates have different control lines, there is at least one \emph{One-Cold} pattern among the set of $n$ patterns that will trigger only one of the two gates, i.e., the only $0$ is applied to a control line of the other gate. Since the other gate is not triggered, the effect induced by the triggered gate can always be passed through the down stream sub-circuit to output $O$, which leads to its detection.
\end{proof}
The \emph{All-1} pattern always trigger both gates. Since these two gates control the same target line, this target line will start from $1$ and is  inverted to $0$ by the first gate, then inverted back to $1$ by the second gate.

Note that there is a third case where both gates have the same control lines and control the same target line. This is a dummy case since both gates are either not triggered or triggered but the second gate always cancels the effect induced by the first gate. While this Trojan may still affect the host's side-channel parameters, we will leave them for future work.

\subsection{Three-Gate Trojans}
For a three-gate Trojan, 5 cases are illustrated in Fig.~\ref{fig:3gate}.

\begin{figure}[!htb]
\centering
\includegraphics[width=0.48\textwidth]{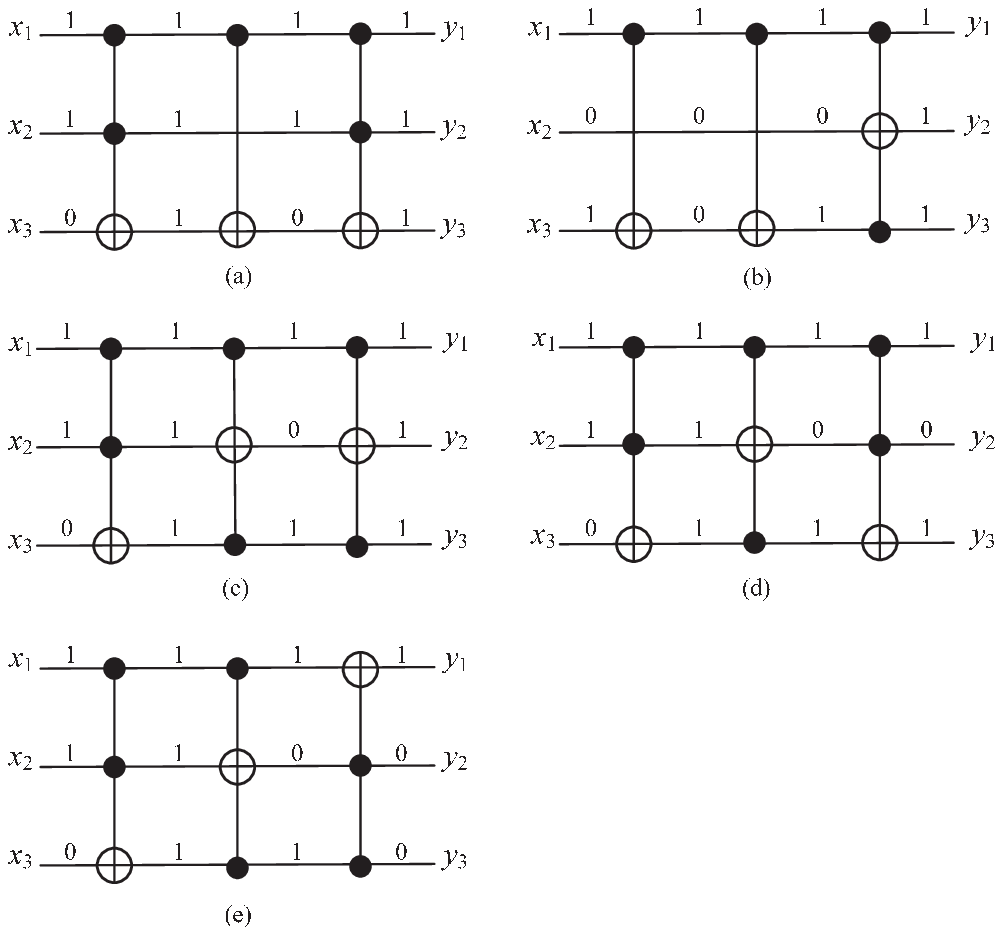}
\caption{All possible three-gate Trojans.}
\label{fig:3gate}
\end{figure}

Case 1: Three gates control the same target line shown in Fig.~\ref{fig:3gate}(a). Consider the \emph{All-1} pattern at $A$ and the \emph{One-Cold} pattern at $A$ where the 0 bit is on the target line. Both patterns will trigger all three gates and the target line will be inverted three times and end up being 0 or 1, respectively. This will result in an unexpected output and hence the Trojan is detected.

Case 2: The former two gates control the same target line. Triggering the third gate or triggering one of the former two gates will guarantee detection of the Trojan. The third gate is triggered by the \emph{All-1} pattern at $A$ or by the \emph{One-Cold} pattern where the 0 bit is on its target line as shown in Fig.~\ref{fig:3gate}(b). When the former two gates have different control lines, one of them can be triggered by applying \emph{One-Cold} patterns; this is Case 2 of 2-gate Trojans. When they have the same control lines, we can ignore them as mentioned in Case 3 of 2-gate Trojans.

Case 3: The last two gates control the same target line as shown in Fig.~\ref{fig:3gate}(c). Triggering the first gate will detect of the Trojan. Both the \emph{All-1} input pattern and the \emph{One-Cold} input pattern at $A$ where the 0 bit is on the target line of the first gate will trigger it.

Case 4: The first and the last gates have the same target line as shown in Fig.~\ref{fig:3gate}(d). This forms a symmetric structure; effects brought by the first gate will be canceled by the third gate when the middle gate is not triggered. The question then is if the middle gate is guaranteed to be detected by using either the \emph{All-1} pattern or the set of \emph{One-Cold} patterns. The answer is yes.
Consider a \emph{One-Cold} pattern at $A$ (which is also the input to the first gate) where the 0 bit is on the target line of the first gate. Since all the other bits are 1, the first gate is triggered and its target line flips $0\rightarrow1$. This results in an \emph{All-1} pattern at the input of the middle gate triggering it. The \emph{All-1} pattern at $A$ will not guarantee triggering of the middle gate. This is because the 1 on the target line of the first gate is inverted to 0 and this 0 may happen to be on the control line of the middle gate.

Case 5: All 3 gates control different target lines as shown in Fig.~\ref{fig:3gate}(e). Triggering any of these gates is enough to detect the Trojan and this can be achieved by either the \emph{All-1} pattern or the set of \emph{One-Cold} patterns.

In summary, the set of \emph{One-Cold} patterns guarantee the detection of any 3-gate Trojans inserted at any position of the host circuit.

\subsection{Four-Gate Trojans}
One can similarly conclude that the set of \emph{One-Cold} patterns guarantees the detection of any 4-gate Trojan inserted at any place in the host circuit.

\subsection{Symmetric Trojans}
For a Trojan with more than four gates, a sophisticated attacker may hide the impact of input patterns with symmetric parts as shown in Fig. \ref{fig:symm}(a).
Unfortunately, neither the \emph{All-1} pattern nor the set of \emph{One-Cold} patterns guarantees detection of such a Trojan.

In Fig.~\ref{fig:symm}(a), a symmetric Trojan is represented as $tTt^{-1}$, where $t$, $T$ and $t^{-1}$ are the left, middle  and the right part of the symmetric Trojan.
The symmetric parts, $t$ and $t^{-1}$, have at least two gates each and have the same set of control lines. They will be triggered by the same set of patterns and control the same set of target lines. The triggering of the symmetric parts $t$ and $t^{-1}$ does not impact the output of the host circuit. Thus such symmetric Trojans can be detected by triggering $T$. However, both patterns fail to do so.
\begin{itemize}
\item When the \emph{All-1} pattern arrives at $A$, $t$ will be triggered and its target lines will be inverted to 0s. Since $t$'s target lines are the control lines to $T$, $T$ won't be able to see all 1s at its input and will remain inactive.
\item If the set of \emph{One-Cold} patterns are applied at $A$, since there is only one 0 bit in the pattern, at most one of the $t$'s target lines will be inverted to 1. $T$ won't see all 1s at its input either and will remain inactive.
\end{itemize}

\begin{figure}[!htb]
\centering
\includegraphics[width=0.45\textwidth]{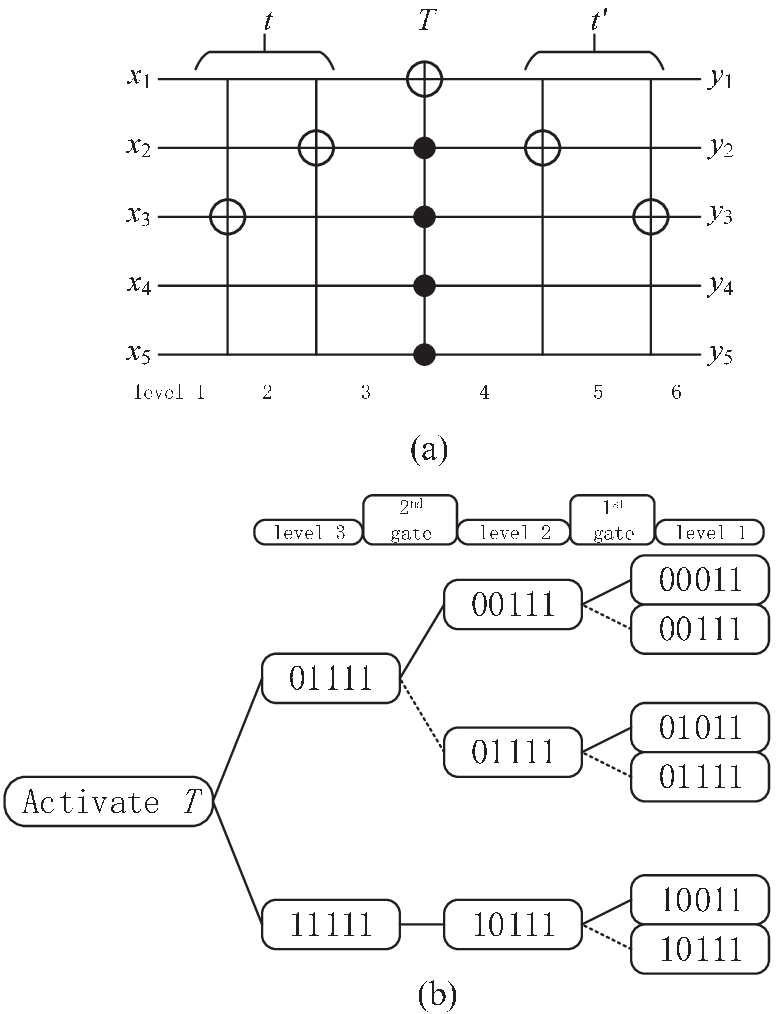}
\caption{Symmetric Trojans. (a) Trojan circuit consists of $tTt^{-1}$; (b) The triggering diagram for $T$; solid and dotted lines represent gates triggered and not triggered respectively. The control lines of $t$ and $t'$ are not shown in the figure.}
\label{fig:symm}
\end{figure}

Fig.~\ref{fig:symm}(b) shows how to trigger $T$. The triggering of $T$ requires the input pattern in level 3 to be $<0,1,1,1,1>$ or $<1,1,1,1,1>$. If we want to apply $<0,1,1,1,1>$ in level 3, the input pattern in level 2 has to be $<0,0,1,1,1>$ when the $2^{nd}$ gate is triggered (solid line) or $<0,1,1,1,1>$ when the $2^{nd}$ gate is not triggered (dotted line). The remaining steps are similarly derived. Each path in the diagram indicates at least one design of a symmetric Trojan. And each leaf of a path represents the pattern that arrives at the Trojan and can eventually activate $T$ part.

As the \emph{All-1} and the set of \emph{One-Cold} patterns cannot guarantee the detection of symmetric patterns, we propose to thwarting hardware Trojans by disabling them in the functional mode. The next section presents the proposed technique.

\section{Detecting and Thwarting Trojans in Reversible Circuits}
\label{sec:proposed}
\subsection{Generalizing to Arbitrary Trojans}

According to the principle of reversibility in Fig. \ref{fig:rtrojan}, a primary input pattern $<i_1, i_2, \cdots, i_n>$ maps correctly to a primary output pattern $<o_1, o_2, \cdots, o_n>$ only if the corresponding intermediate result $<a_1, a_2, \cdots, a_n>$ is the same as $<b_1, b_2, \cdots, b_n>$, i.e., the Trojan passes its input unchanged. Primary input patterns that make an inconsistent mapping between $<a_1, a_2, \cdots, a_n>$ and $<b_1, b_2, \cdots, b_n>$ trigger the Trojan, resulting in unexpected primary outputs.
Since a Trojan is reversible by itself, the size of its truth table is the same as its host circuit.
A Trojan should not be triggered for all the entries in its truth table where the input equals the output (\emph{equal I/O pairs}).

We can use $D=\frac{Num(equal~I/O~pairs)}{Num(all~I/O~pairs)} \leq 1$ to represent the difficulty of triggering a Trojan. The larger the D, the more difficult it is to detect/trigger the Trojan when random patterns are applied.

Fig.~\ref{fig:example} shows reversible circuits with a Trojan gate in dotted rectangle. In Fig.~\ref{fig:example}(a), the Trojan gate has a single control and target line. When the control line is set to 1, the Trojan is activated and $D=0.5$.
In Fig.~\ref{fig:example}(b), the Trojan gate has 2 control lines and 1 target line. It will be activated only when both its control lines receive 1. This translates to 6-out-of-8 I/O pairs and results in $D=0.75$.
The truth tables of the two Trojans are shown in Fig.~\ref{fig:example}(c), where triggering I/O pairs are in bold.
\begin{figure}[htbp]
\centering
\includegraphics[width=0.4\textwidth]{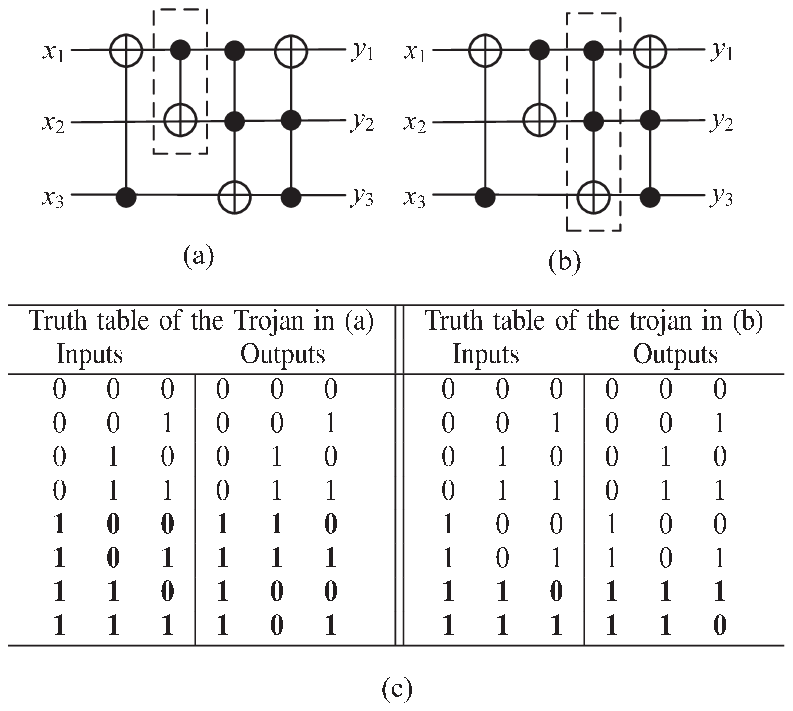}
\caption{Reversible circuits with an arbitrary Trojan gate in dotted rectangle.}
\label{fig:example}
\end{figure}

\subsection{Detecting Trojans in Reversible Circuits w/o Ancillary Inputs}
In inherently reversible circuits, the probability that each line is a 1 (or 0) is 0.5. Hence applying random test patterns at primary inputs can be a basic detection method since the randomness at the primary inputs will be passed to the Trojan's input no matter where it is inserted. Hence the difficulty/chance of detecting a random Trojan in inherently reversible circuits is directly related to the $D$ of this Trojan.

\subsection{Detecting Trojans in Reversible Circuits w/ Ancillary Inputs}
Non-reversible functions, such as add and multiply, can be embedded into a reversible circuit. This is a more popular case since most functions are not reversible.
The embedding process first adds garbage outputs to map each possible input to a unique output, and then adds ancillary inputs to make the number of inputs equal to the number of outputs. Ancillary inputs can be assigned random values just like primary inputs when the reversible circuit is under testing and assigned a constant value (0~or~1) in the functional mode.

Let us use the half adder example. The truth table of a half adder embedded in a reversible circuit is shown in TABLE \ref{tab:rehalf_2}. In the test mode, all entries of the truth table can be exercised. In the functional mode, only the bold entries are visited.
If a triggering pattern of a Trojan is one of these never-reached entries, it will never be received by the Trojan in the functional mode. If all triggering patterns of a Trojan are the never-reached entries, the Trojan is never triggered in the functional mode.
This is illustrated in Fig. \ref{fig:diff}. Each gate is annotated with its triggering probability. The ancillary inputs result in a 0 triggering probability of the fourth gate in Fig. \ref{fig:diff}(b).

\textbf{This property presents a  dilemma to a Trojan designer.
On one hand, to avoid being triggered and hence detected in the test mode where all inputs, including both primary and ancillary inputs could be exercised for testing, Trojan designers have to minimize the number of patterns that can trigger their Trojans. That is, make  $D$ close to 1. On the other hand, the less the number of triggering patterns, the higher the chance they are never exercised in the functional mode where ancillary inputs will only take fixed values.}
This can be exploited by the designers of host circuits to (partially or completely) disable Trojans in the functional mode.

\begin{figure}[htbp]
\centering
\includegraphics[width=0.42\textwidth]{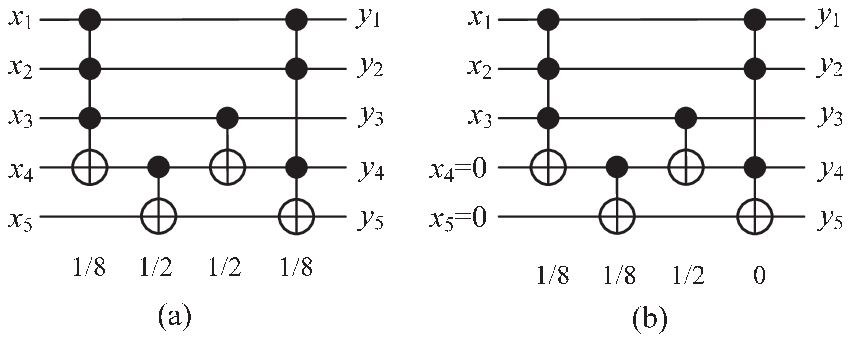}
\caption{Triggering a Trojan in reversible circuits (a) without ancillary inputs and (b) with ancillary inputs.}
\label{fig:diff}
\end{figure}

\subsection{Scramble Inputs and Outputs to Thwart Trojans}
\label{sec:pospro}
Consider scrambling of inputs. The scrambling can be achieved in two ways. First, one can scramble the inputs to the original function by exploiting the existing ancillary inputs. Second, introduce ancillary inputs to improve scrambling. The designer scrambles in the design phase and keeps it secret from the foundry, as shown in Fig.~\ref{fig:chain2}. Only authorized users  obtain and use the scrambling code (the assignments to ancillary inputs).
\begin{figure}[htbp]
\centering
\includegraphics[width=0.46\textwidth]{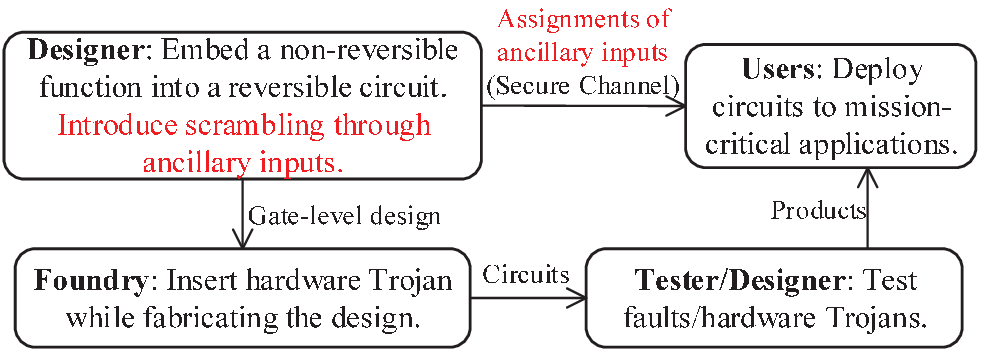}
\caption{Designer scrambles I/O pairs through ancillary inputs.}
\label{fig:chain2}
\end{figure}
\subsubsection{Exploit Ancillary Inputs}

For simplicity, ancillary inputs are usually assigned all 0 or all 1 during functional mode. The truth table of \emph{decod24\_10} with two ancillary inputs being 0 is shown in TABLE \ref{tab:assign_1}. This assignment is easy to crack by an attacker since only two guesses are necessary.
Scrambling can be more sophisticated, since the ancillary inputs do not have to be a fixed constants. They can be designed to take any values in the functional mode, as long as the resulting truth table is reversible.

Starting with a reversible circuit with $p$ primary inputs and $c$ ancillary inputs that take all 0 (or all 1) for all
primary input patterns, the first thing we can do is to make some of the ancillary inputs 0 and the rest 1. One example is shown in TABLE \ref{tab:assign_2}. The number of guesses for the values taken by ancillary inputs will increase up to $2^c$.

To further decrease successful guessing, the values of ancillary inputs can be randomly chosen for each input pattern, as long as these input patterns map to correct output patterns. One example is shown in TABLE \ref{tab:assign_3}. This makes the number of guesses up to $2^{p+c}$.
Additional ancillary inputs can be added to improve scrambling, at extra cost.

\begin{table}
\caption{Assignments of ancillary inputs of \emph{decod\_24\_10}.}
\label{tab:halfrev}
\resizebox{0.51\textwidth}{0.035\textwidth}{
\subtable[The assignment of all ancillary inputs being 0]{
       \begin{tabular}{c|c}
        $x_1x_2a_1a_2$ &$y_1y_2y_3y_4$\\\hline
        \textbf{0 0 0 0} & \textbf{0 0 0 1}\\
        0 0 0 1 & 0 0 0 0\\
        0 0 1 0 & 0 1 1 0\\
        0 0 1 1 & 0 1 1 1\\
        \textbf{0 1 0 0} & \textbf{0 0 1 0}\\
        0 1 0 1 & 0 1 0 1\\
        0 1 1 0 & 1 1 0 0\\
        0 1 1 1 & 0 0 1 1\\
        \textbf{1 0 0 0} & \textbf{0 1 0 0}\\
        1 0 0 1 & 1 0 0 1\\
        1 0 1 0 & 1 1 1 0\\
        1 0 1 1 & 1 1 1 1\\
        \textbf{1 1 0 0} & \textbf{1 0 0 0}\\
        1 1 0 1 & 1 1 0 1\\
        1 1 1 0 & 1 0 1 0\\
        1 1 1 1 & 1 0 1 1\\
       \end{tabular}
       \label{tab:assign_1}
}
\quad
\subtable[The assignment of half of ancillary inputs being 0]{
       \begin{tabular}{c|c}
        $x_1x_2a_1a_2$ &$y_1y_2y_3y_4$\\\hline
        0 0 0 0 & 0 1 1 0\\
        0 0 0 1 & 0 0 0 0\\
        \textbf{0 0 1 0} & \textbf{0 0 0 1}\\
        0 0 1 1 & 0 1 1 1\\
        0 1 0 0 & 1 1 0 0\\
        0 1 0 1 & 0 1 0 1\\
        \textbf{0 1 1 0} & \textbf{0 0 1 0}\\
        0 1 1 1 & 0 0 1 1\\
        1 0 0 0 & 1 1 1 0\\
        1 0 0 1 & 1 0 0 1\\
        \textbf{1 0 1 0} & \textbf{0 1 0 0}\\
        1 0 1 1 & 1 1 1 1\\
        1 1 0 0 & 1 0 1 0\\
        1 1 0 1 & 1 1 0 1\\
        \textbf{1 1 1 0} & \textbf{1 0 0 0}\\
        1 1 1 1 & 1 0 1 1\\
       \end{tabular}
       \label{tab:assign_2}
}
\quad
\subtable[The assignment of random ancillary inputs]{
       \begin{tabular}{c|c}
        $x_1x_2a_1a_2$ &$y_1y_2y_3y_4$\\\hline
        0 0 0 0 & 0 0 0 0\\
        \textbf{0 0 0 1} & \textbf{0 0 0 1}\\
        0 0 1 0 & 0 1 1 0\\
        0 0 1 1 & 0 1 1 1\\
        0 1 0 0 & 1 1 01 0\\
        0 1 0 1 & 0 1 0 1\\
        \textbf{0 1 1 0} & \textbf{0 0 1 0}\\
        0 1 1 1 & 0 0 1 1\\
        1 0 0 0 & 1 1 1 1\\
        1 0 0 1 & 1 0 0 1\\
        1 0 1 0 & 1 1 1 0\\
        \textbf{1 0 1 1} & \textbf{0 1 0 0}\\
        1 1 0 0 & 1 1 0 1\\
        \textbf{1 1 0 1} & \textbf{1 0 0 0}\\
        1 1 1 0 & 1 0 1 0\\
        1 1 1 1 & 1 0 1 1\\
       \end{tabular}
       \label{tab:assign_3}
}}
\end{table}

\subsubsection{Probability that a Trojan is Never Triggered}
Assume a $p$-input non-reversible function is embedded to a $n$-input reversible circuit with $c$ ancillary inputs ($n=p+c$).
The attacker sees a circuit whose truth table has $2^n=2^{p+c}$ entries. He then designs a Trojan that has $t$ triggering patterns. $t$ is chosen to be small enough to pass testing using random patterns picked from total $2^{p+c}$ patterns. However, during functional mode, all or part of the $t$ patterns may never be reached. The probability is calculated as follows.
\begin{itemize}
    \item the probability that no input patterns trigger the Trojan is $\frac{{2^{p}\choose0}{{{2^{p+c}-2^{p}}\choose {t}}}}{{2^{p+c}\choose{t}}}$, i.e., all $t$ patterns are never reached.
    \item the probability that 1 input pattern triggers the Trojan is $\frac{{{2^{p}}\choose1}*{{2^{p+c}-2^{p}}\choose {t-1}}}{{{2^{p+c}}\choose {t}}}$, i.e., other $t-1$ patterns are never reached.
    \item the probability that 2 input patterns trigger the Trojan is $\frac{{{2^{p}}\choose2}*{{2^{p+c}-2^{p}}\choose {t-2}}}{{{2^{p+c}}\choose {t}}}$;
    \item $\cdots$
    \item the probability that $i$ input patterns trigger the Trojan is $\frac{{{2^{p}}\choose i}*{{2^{p+c}-2^{p}}\choose {t-i}}}{{{2^{p+c}}\choose {t}}}$;
\end{itemize}

Fig. \ref{fig:protri} illustrates the effect of ancillary inputs on triggering a Trojan. The number of patterns that can trigger the Trojan is set to $t=8$. The number of primary inputs and ancillary inputs is varied from 3 to 20 and from 1 to 20, respectively. The probability that the Trojan is never triggered is $\sim$ 0 when there are few ancillary inputs no matter the number of primary inputs. A Trojan in a reversible circuit can always be triggered in such a situation. However, this probability increases exponentially with the number of ancillary inputs. A small number of ancillary inputs can disable triggering a Trojan with a very high probability. As a result, extra ancillary inputs can protect the host from Trojan even when there are insufficient number of ancillary inputs created during embedding.
\begin{figure}[htbp]
\centering
\includegraphics[width=0.42\textwidth]{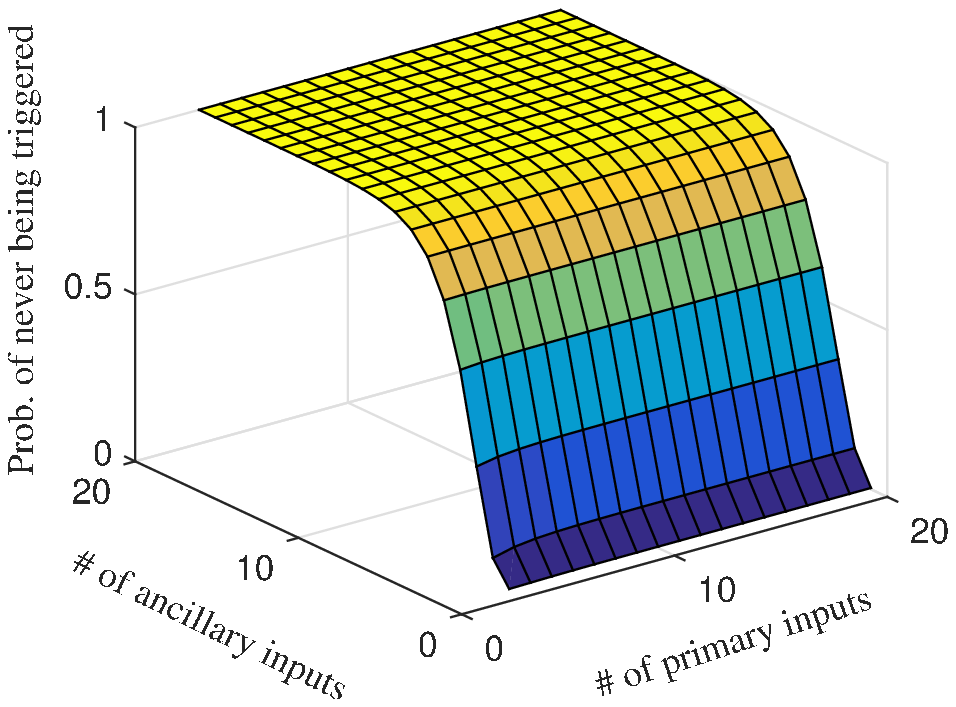}
\caption{The probability that a Trojan is never triggered as a function the number of primary inputs and ancillary inputs.}
\label{fig:protri}
\end{figure}


\section{Experiments}
\label{sec:experiments}

The experiments use the QMDD synthesis approach that takes as input a reversible truth table, outputs its implementation in reversible logic, and reports the cost \cite{ZulehnerW17Emb,6165069}. Benchmark circuits are taken from \emph{revlib} \cite{WGT+:2008}.
The scrambling technique in Section \ref{sec:proposed} is partitioned into several levels, including all ancillary inputs being 0 ($baseline$), half ancillary inputs being 0 and the other half being 1 ($Lv1$), random ancillary inputs for each primary input pattern ($Lv2$), and adding extra 1 ancillary input ($Lv3$).

To have a fair comparison on the probability of disabling Trojans across different scrambling levels, the study is carried on the Trojans that have only 1 triggering pattern. The Trojans with this assumption always have a minimum probability being detected in the test mode, which is favored by the attacker.
Assuming random inputs patterns are applied, we can then calculate the probability of such Trojans being disabled in the functional mode.
Since scrambling levels $baseline$, $Lv1$ and $Lv2$ have the same number of ancillary inputs, they have the same probability of disabling same Trojans.
After adding an extra ancillary input, the probability will be improved. Experimental results of $baseline-Lv2$, $Lv3$ are shown in TABLE \ref{tab:pro}.

From TABLE \ref{tab:pro}, we can observe that, without adding any extra ancillary inputs, all benchmarks have at least 50\% chance of disabling the Trojans that have just 1 triggering pattern. Some benchmarks are even close to 100\%. With one more extra ancillary inputs, the probability is always higher.
This is great in thwarting Trojans in reversible circuits since the attackers may need to re-evaluate the trade-off between the great risk they will bear by inserting these Trojans and the usefulness of these Trojans.

\begin{table}[!htbp]
\centering
\caption{The probability that the Trojan is disabled in the functional mode.}
\label{tab:pro}
\renewcommand{\arraystretch}{1.1}
\resizebox{0.5\textwidth}{0.32\textwidth}{
\begin{tabular}{l|l|l|ll}
\hline
\multirow{2}{*}{Benchmark} & \multirow{2}{*}{Total inputs} & \multirow{2}{*}{constants} & \multicolumn{2}{c}{Probability of being disabled (\%)} \\
               &    &   & Baseline-Lv2    & Lv3     \\\hline
decod24\_10 &4  &2  &75.0 &87.5    \\
4gt12\_24   &5  &1  &50.0 &75.0 \\
decod24\_32 &6  &3  &87.5 &93.8 \\
f2\_158     &7  &3  &87.5 &93.8  \\
rd53\_68    &7  &2  &75.0 &87.5  \\
rd73\_69    &9  &2  &75.0 &87.5 \\
sqn\_203    &9  &2  &75.0 &87.5 \\
squar5\_206 &9  &4  &93.8 &96.9 \\
wim\_220    &9  &5  &96.9 &98.4 \\
dc1\_142    &10 &6  &98.4 &99.2  \\
dist\_144   &10 &2  &75.0 &87.5\\
root\_197   &10 &2  &75.0 &87.5\\
clip\_124   &11 &2  &75.0 &87.5\\
rd84\_70    &11 &3  &87.5 &93.8 \\
sqr6\_204   &12 &6  &98.4 &99.2\\
cm42a\_125  &13 &9  &99.8 &99.9 \\
dc2\_143    &13 &5  &96.9 &98.4 \\
mlp4\_184   &13 &5  &96.9 &98.4\\
pm1\_192    &13 &9  &99.8 &99.9\\
alu2\_96    &14 &4  &93.8 &96.9 \\
alu3\_97    &14 &4  &93.8 &96.9\\
example2\_156&14&4  &93.8 &96.9\\
inc\_170    &14 &7  &99.2 & 99.6\\
misex1\_178 &14 &6  &98.4 &99.2 \\
sao2\_199   &14 &4  &93.8 &96.9 \\
co14\_135   &15 &1  &50.0 &75.0  \\
dk27\_146   &15 &6  &98.4 &99.2 \\
x2\_223     &16 &6  &98.4 &99.2\\
t481\_208   &17 &1  &50.0 &75.0 \\\hline
Average     &   &   &85.8 &92.9\\\hline
\end{tabular}}
\end{table}

The overhead of all scrambling levels is also evaluated for each benchmark circuit in terms of the number of lines (Line Cost), the number of gates (Gate Cost) and the number of primitive reversible logic gates (1*1, 2*2) required to realize the circuit (Quantum Cost). A Toffoli gate may consist of several primitive reversible logic gates, so the quantum cost is several times larger than the gate cost. The overhead of each scrambling level ($Lv1-3$) is compared against the baseline. For reversible circuits that only have 1 ancillary input, the overhead of $Lv1$ is not calculated in Table \ref{tab:cost}.

It is observed that for the scrambling levels that has no extra ancillary input ($baseline$-$Lv2$), the line cost is always 0. For the scrambling levels that have extra ancillary input(s) ($Lv3$), the line cost presents, but the relative overhead decreases with the scale of benchmark circuits.

Besides, $Lv2$ has maximum average gate cost and quantum cost compared to other scrambling levels. This is because, assigning random ancillary inputs for each input pattern forms a most complex function, which will need more gates to realize it.
As a result, for both security- and gate/quantum cost-aware designs, introducing extra ancillary inputs can be a substitution for further improving scrambling levels with less gate/quantum cost while at extra line cost.
It is also observed the gate cost and quantum cost increase significantly for some designs. This is partially because the state-of-art synthesis algorithm is still on its early stage where obtaining a valid implementation is still the primary goal.

\begin{table*}[!htbp]
\centering
\caption{Overhead evaluation of the proposed input-scrambling scheme.}
\label{tab:cost}
\renewcommand{\arraystretch}{1.2}
\resizebox{0.68\textwidth}{0.33\textwidth}{
\begin{tabular}{l|lll|lll|lll}
\hline
\multirow{2}{*}{Benchmark}  & \multicolumn{3}{c|}{Line Cost (\%)} & \multicolumn{3}{c|}{Gate Cost (\%)} & \multicolumn{3}{c}{Quantum Cost (\%)} \\
               &Lv1&Lv2&Lv3&Lv1&Lv2&Lv3&Lv1&Lv2&Lv3\\\hline
decod24\_10       &0 &0  &25       &44     &0&33.3
&54&1.2&88\\
4gt12\_24         &N/A    &0  &20        &N/A      &100&325
   &N/A     &168 &782 \\
decod24\_32      &0  &0  &17       &20    &57&-17
   &19   &77 &-15  \\
f2\_158         &0  &0  &14       &1.4     &81&-35
   &0.6   &96 &-29 \\
rd53\_68         &0  &0  &14       &31     &97&44
   &30   &120 &59  \\
rd73\_69         &0  &0  &11      &-15   &38&30
  &-16  &45 &36  \\
sqn\_203         &0  &0  &11       &-0.5   &55&17
  &-2  &54 &20 \\
squar5\_206      &0  &0  &11         &96   &344&56
   &125  &496 &77  \\
wim\_220         &0  &0 &11       &118   &387&81
   &133  &526 &89 \\
dc1\_142       &0  &0 &10        &98   &428&141
   &112  &629 &189  \\
dist\_144   &0  &0 &10       &28  &48&77
 &29 &51 &79 \\
root\_197       &0  &0 &10       &45  &288&72
  &56  &369 &86  \\
clip\_124       &0  &0 &9        &235  &241&153
  &277 &294 &170\\
rd84\_70         &0  &0 &9        &15   &163&50
 &14 &189 &62 \\
sqr6\_204       &0  &0 &8       &220   &1356&166
   &538  &3698 &353  \\
cm42a\_125      &0  &0 &8         &44    &953&444
   &116   &3347 &1018 \\
dc2\_143         &0  &0 &8         &192   &5257 &529
   &457&	20691&	2129	 \\
mlp4\_184    &0  &0 &8      &145  &483&-3
  &159&	620&	-1.3	\\
pm1\_192     &0  &0 &8       &44    &1109&371
   &116&	3495&	723	\\
alu2\_96     &0  &0 &7       &43   &1977&97
  &49&	3124&	177	\\
alu3\_97     &0  &0 &7       &115   &15922&276
  &168&	32454&	618\\
example2\_156   &0  &0 &7       &43   &1847&97
  &49&	2913&	177\\
inc\_170     &0  &0 &7      &68    &8480&761
  &103&	33166&	2481\\
misex1\_178     &0  &0 &7       &526   &9725&817
   &1122&	28870&	1886\\
sao2\_199       &0  &0 &7        &36   &8054&655
  &18&	17197&	1439\\
dk27\_146    &0  &0 &7     &99 &3854&240
 &124&	5149&	320\\
x2\_223          &0  &0 &6     &127  &5169&-19
 &124&	5472&	-18\\
t481\_208    &N/A&0 &4     &N/A  &3156&716
 &N/A&	8014&	1363\\\hline
AVERAGE      &0  &0 &10     &91  &2488&221
&148 &6118&	513\\
\hline
\end{tabular}}
\end{table*}

\section{Conclusion and Future Work}
\label{sec:conclusion}
As the circuit designer has a good knowledge of all gates in reversible circuits, any missing control/gate fault caused by imperfect fabrication is easy to detect through test pattern generation.
Compared to existing fault models, deliberately designed and inserted Trojans are sophisticated and unpredictable, and hence more difficult to detect and counter.

Our investigation shows that, by exploiting the inherent reversibility of these circuits, a set of predefined test patterns can trigger (and hence detect) any Trojan up-to 4 gates in size inserted at any position in the host circuit. The attacker is forced to insert symmetric Trojans that have 5 or more gates. We propose to scramble inputs to thwart such Trojans in the functional mode of host circuits. The inputs to the circuit are scrambled by assigning secret values to the ancillary inputs and introducing extra ancillary inputs. Without knowing these secret values and extra ancillary inputs, a Trojan designer has to risk an increased chance that the Trojan may never be triggered during circuit functional mode, which invalidates his original motivation.

\subsection{Discussion and Future Work}
It is observed in the experiments that the overhead varies significantly for benchmark circuits within the same scrambling level or for the same circuits at different scrambling levels. That is, the increase of overhead (in terms of gate cost or quantum cost) from $baseline$ to $Lv3$ are not in a fixed relation. This is because, the current synthesis approaches do not take the scrambling technique into consideration. Thus developing an overhead- and security-aware synthesis approach will be an important future work.

Besides, this work only considers hardware Trojans that consist of consecutive gates. However, an attacker may design a Trojan that has several distributed parts and each part may consist of multiple gates. The investigation of this kind of Trojan in reversible circuits will also be a part of future work.
\bibliographystyle{IEEEtran}
\bibliography{mybib}

\begin{thebibliography}{10}
\providecommand{\url}[1]{#1}
\csname url@samestyle\endcsname
\providecommand{\newblock}{\relax}
\providecommand{\bibinfo}[2]{#2}
\providecommand{\BIBentrySTDinterwordspacing}{\spaceskip=0pt\relax}
\providecommand{\BIBentryALTinterwordstretchfactor}{4}
\providecommand{\BIBentryALTinterwordspacing}{\spaceskip=\fontdimen2\font plus
\BIBentryALTinterwordstretchfactor\fontdimen3\font minus
  \fontdimen4\font\relax}
\providecommand{\BIBforeignlanguage}[2]{{%
\expandafter\ifx\csname l@#1\endcsname\relax
\typeout{** WARNING: IEEEtran.bst: No hyphenation pattern has been}%
\typeout{** loaded for the language `#1'. Using the pattern for}%
\typeout{** the default language instead.}%
\else
\language=\csname l@#1\endcsname
\fi
#2}}
\providecommand{\BIBdecl}{\relax}
\BIBdecl

\bibitem{2000:QCQ:544199}
\emph{Quantum Computation and Quantum Information}.\hskip 1em plus 0.5em minus
  0.4em\relax New York, NY, USA: Cambridge University Press, 2000.

\bibitem{BBC+:95}
A.~Barenco, C.~H. Bennett, R.~Cleve, D.~DiVinchenzo, N.~Margolus, P.~Shor,
  T.~Sleator, J.~Smolin, and H.~Weinfurter, ``Elementary gates for quantum
  computation,'' \emph{The American Physical Society}, vol.~52, pp. 3457--3467,
  1995.

\bibitem{DBLP:conf/rc/NiemannBCJW15}
P.~Niemann, S.~Basu, A.~Chakrabarti, N.~K. Jha, and R.~Wille, ``Synthesis of
  quantum circuits for dedicated physical machine descriptions,'' 2015, pp.
  248--264.

\bibitem{Gro:96}
L.~K. Grover, ``A fast quantum mechanical algorithm for database search,'' in
  \emph{Theory of computing}, 1996, pp. 212--219.

\bibitem{Sho:94}
P.~W. Shor, ``Algorithms for quantum computation: discrete logarithms and
  factoring,'' \emph{Foundations of Computer Science}, pp. 124--134, 1994.

\bibitem{SAZS:2011}
M.~Saeedi, M.~Arabzadeh, M.~S. Zamani, and M.~Sedighi, ``Block-based
  quantum-logic synthesis,'' vol.~11, no. 3{\&}4, pp. 262--277, 2011.

\bibitem{DBLP:conf/ismvl/DrechslerW11}
\BIBentryALTinterwordspacing
R.~Drechsler and R.~Wille, ``From truth tables to programming languages:
  Progress in the design of reversible circuits,'' in \emph{International
  Symposium on Multiple-Valued Logic, {ISMVL}}, 2011, pp. 78--85. [Online].
  Available: \url{http://dx.doi.org/10.1109/ISMVL.2011.40}
\BIBentrySTDinterwordspacing

\bibitem{GWDD:2009b}
D.~Gro{\ss}e, R.~Wille, G.~W. Dueck, and R.~Drechsler, ``Exact synthesis of
  elementary quantum gate circuits,'' vol.~15, no.~4, pp. 270--275, 2009.

\bibitem{Landauer61}
R.~Landauer, ``Irreversibility and heat generation in the computing process,''
  \emph{IBM Journal of Research and Development}, vol.~5, no.~3, pp. 183--191,
  1961.

\bibitem{Ben:73}
C.~H. Bennett, ``Logical reversibility of computation,'' \emph{IBM J. Res.
  Dev}, vol.~17, no.~6, pp. 525--532, 1973.

\bibitem{BAP+:2012}
A.~Berut, A.~Arakelyan, A.~Petrosyan, S.~Ciliberto, R.~Dillenschneider, and
  E.~Lutz, ``Experimental verification of {L}andauer's principle linking
  information and thermodynamics,'' \emph{Nature}, vol. 483, pp. 187--189,
  2012.

\bibitem{363692}
W.~C. Athas and L.~J. Svensson, ``Reversible logic issues in adiabatic cmos,''
  in \emph{Physics and Computation, 1994. PhysComp '94, Proceedings., Workshop
  on}, Nov 1994, pp. 111--118.

\bibitem{DBLP:conf/date/AmaruGW16}
\BIBentryALTinterwordspacing
L.~G. Amar{\`{u}}, P.~Gaillardon, R.~Wille, and G.~D. Micheli, ``Exploiting
  inherent characteristics of reversible circuits for faster combinational
  equivalence checking,'' 2016, pp. 175--180. [Online]. Available:
  \url{http://ieeexplore.ieee.org/xpl/freeabs_all.jsp?arnumber=7459300}
\BIBentrySTDinterwordspacing

\bibitem{Cuykendall:87}
\BIBentryALTinterwordspacing
R.~Cuykendall and D.~R. Andersen, ``Reversible optical computing circuits,''
  \emph{Opt. Lett.}, vol.~12, no.~7, pp. 542--544, Jul 1987. [Online].
  Available: \url{http://ol.osa.org/abstract.cfm?URI=ol-12-7-542}
\BIBentrySTDinterwordspacing

\bibitem{roy2011all}
S.~Roy, P.~Sethi, J.~Topolancik, and F.~Vollmer, ``All-optical reversible logic
  gates with optically controlled bacteriorhodopsin protein-coated
  microresonators,'' \emph{Advances in Optical Technologies}, vol. 2012, 2011.

\bibitem{taha2016fundamentals}
S.~M.~R. Taha, ``Fundamentals of reversible logic,'' in \emph{Reversible Logic
  Synthesis Methodologies with Application to Quantum Computing}.\hskip 1em
  plus 0.5em minus 0.4em\relax Springer, 2016, pp. 7--16.

\bibitem{SPMH:2002}
V.~V. Shende, A.~K. Prasad, I.~L. Markov, and J.~P. Hayes, ``Reversible logic
  circuit synthesis,'' 2002, pp. 353--360.

\bibitem{MillerMD03}
D.~M. Miller, D.~Maslov, and G.~W. Dueck, ``A transformation based algorithm
  for reversible logic synthesis,'' in \emph{Proc.\ of the 40th Design
  Automation Conference, DAC 2003}.\hskip 1em plus 0.5em minus 0.4em\relax ACM,
  2003, pp. 318--323.

\bibitem{YSHP:2005}
G.~Yang, X.~Song, W.~N.~N. Hung, and M.~A. Perkowski, ``Fast synthesis of exact
  minimal reversible circuits using group theory,'' 2005, pp. 1002--1005.

\bibitem{GAJ:2006}
P.~Gupta, A.~Agrawal, and N.~K. Jha, ``An algorithm for synthesis of reversible
  logic circuits,'' vol.~25, no.~11, pp. 2317--2330, 2006.

\bibitem{wille2009bdd}
R.~Wille and R.~Drechsler, ``Bdd-based synthesis of reversible logic for large
  functions,'' in \emph{Proceedings of the 46th Annual Design Automation
  Conference}.\hskip 1em plus 0.5em minus 0.4em\relax ACM, 2009, pp. 270--275.

\bibitem{SWH+:2012}
M.~Soeken, R.~Wille, C.~Hilken, N.~Przigoda, and R.~Drechsler, ``Synthesis of
  reversible circuits with minimal lines for large functions,'' 2012, pp.
  85--92.

\bibitem{ZulehnerW17Emb}
A.~Zulehner and R.~Wille, ``Make it reversible: Efficient embedding of
  non-reversible functions,'' 2017.

\bibitem{Jin2009Hardware}
Y.~Jin, N.~Kupp, and Y.~Makris, ``Experiences in hardware trojan design and
  implementation,'' in \emph{HOST'09}.

\bibitem{tehranipoor2010survey}
M.~Tehranipoor and F.~Koushanfar, ``A survey of hardware trojan taxonomy and
  detection,'' \emph{Design \& Test of Computers, IEEE}, 2010.

\bibitem{xiao2013bisa}
K.~Xiao and M.~Tehranipoor, ``Bisa: Built-in self-authentication for preventing
  hardware trojan insertion,'' in \emph{Hardware-Oriented Security and Trust
  (HOST), 2013 IEEE International Symposium on}.\hskip 1em plus 0.5em minus
  0.4em\relax IEEE, 2013, pp. 45--50.

\bibitem{cui2014high}
X.~Cui, K.~Ma, L.~Shi, and K.~Wu, ``High-level synthesis for run-time hardware
  trojan detection and recovery,'' in \emph{Design Automation Conference (DAC),
  2014 51st ACM/EDAC/IEEE}.\hskip 1em plus 0.5em minus 0.4em\relax IEEE, 2014,
  pp. 1--6.

\bibitem{salmani2012anovel}
H.~Salmani, M.~Tehranipoor, and J.~Plusquellic, ``A novel technique for
  improving hardware trojan detection and reducing trojan activation time,''
  \emph{VLSI Systems, IEEE Transactions on}, vol.~20, no.~1, pp. 112--125, Jan
  2012.

\bibitem{wang2008hardware}
X.~Wang, H.~Salmani, M.~Tehranipoor, and J.~Plusquellic, ``Hardware trojan
  detection and isolation using current integration and localized current
  analysis,'' in \emph{DFTVS'08}, pp. 87--95.

\bibitem{jin2008hardware}
Y.~Jin and Y.~Makris, ``Hardware trojan detection using path delay
  fingerprint,'' in \emph{HOST'08. IEEE International Workshop on}, pp. 51--57.

\bibitem{7390050}
H.~Salmani and M.~M. Tehranipoor, ``Vulnerability analysis of a circuit layout
  to hardware trojan insertion,'' \emph{IEEE Transactions on Information
  Forensics and Security}, vol.~11, no.~6, pp. 1214--1225, June 2016.

\bibitem{patel2004fault}
K.~N. Patel, J.~P. Hayes, and I.~L. Markov, ``Fault testing for reversible
  circuits,'' \emph{IEEE Transactions on Computer-Aided Design of Integrated
  Circuits and Systems}, vol.~23, no.~8, pp. 1220--1230, 2004.

\bibitem{ramasamy2004fault}
K.~Ramasamy, R.~Tagare, E.~Perkins, and M.~Perkowski, ``Fault localization in
  reversible circuits is easier than for classical circuits,'' 2004.

\bibitem{hayes2004testing}
J.~P. Hayes, I.~Polian, and B.~Becker, ``Testing for missing-gate faults in
  reversible circuits,'' in \emph{Test Symposium, 2004. 13th Asian}.\hskip 1em
  plus 0.5em minus 0.4em\relax IEEE, 2004, pp. 100--105.

\bibitem{wille2011atpg}
R.~Wille, H.~Zhang, and R.~Drechsler, ``Atpg for reversible circuits using
  simulation, boolean satisfiability, and pseudo boolean optimization,'' in
  \emph{VLSI (ISVLSI), 2011 IEEE Computer Society Annual Symposium on}.\hskip
  1em plus 0.5em minus 0.4em\relax IEEE, 2011, pp. 120--125.

\bibitem{6165069}
M.~Soeken, R.~Wille, C.~Hilken, N.~Przigoda, and R.~Drechsler, ``Synthesis of
  reversible circuits with minimal lines for large functions,'' in \emph{17th
  Asia and South Pacific Design Automation Conference}, Jan 2012, pp. 85--92.

\bibitem{WGT+:2008}
R.~Wille, D.~Gro{\ss}e, L.~Teuber, G.~W. Dueck, and R.~Drechsler, ``{RevLib:}
  an online resource for reversible functions and reversible circuits,'' 2008,
  pp. 220--225, {RevLib} is available at http://www.revlib.org.

\end{thebibliography}

\end{document}